\documentclass[]{IEEEtran}
\usepackage{graphicx,subfigure,epsfig,amsmath,amsbsy,amssymb,citesort,verbatim,url,algorithm,algpseudocode,soul}

\usepackage[usenames]{color}
\usepackage{setspace,amsfonts,epsf,url,bigints}
\usepackage{bbm}

\renewcommand{\baselinestretch}{1}

\def \A {\mathbf{A}}
\def \X {\mathbf{X}}
\def \x {\mathbf{x}}
\def \y {\mathbf{y}}

\def \s {\mathbf{s}}
\def \v {\mathbf{v}}
\def \z {\mathbf{z}}

\def \h {\mathbf{h}}
\def \l {\left}
\def \r {\right}

\def \w {\mathbf{w}}
\def \a {\mathbf{a}}

\def \q {\mathbf{q}}

\def \k {\mathbf{k}}

\newcommand{\sij}[2] {\sum_{#1}^{#2}}

\def \e {\mathbf{e}}
\def \l {\left}
\def \r {\right}

\usepackage[normalem]{ulem}

\newtheorem{myDef}{Definition}
\newtheorem{myLemma}{Lemma}

\newtheorem{myConj}{Conjecture}

\begin{document}

\title{Performance Limits with Additive Error Metrics in Noisy Multi-Measurement Vector Problems}
\author{Junan Zhu,~\IEEEmembership{Member,~IEEE} and
Dror~Baron,~\IEEEmembership{Senior Member,~IEEE}
\thanks{The work was supported in part by the
National Science Foundation under the Grants
CCF-1217749 and ECCS-1611112.}
\thanks{Junan Zhu is with Bloomberg L.P., New York, NY 10017, and Dror Baron is with the Department of Electrical and Computer Engineering, NC State University, Raleigh, NC 27695.
E-mail: \{jzhu9, barondror\}@ncsu.edu.}}
\maketitle

\begin{abstract}
Real-world applications such as magnetic resonance imaging with multiple coils, multi-user communication, and diffuse optical tomography often assume a linear model where several sparse signals sharing common sparse supports  are acquired by several measurement matrices and then contaminated by noise.
Multi-measurement vector (MMV) problems consider the estimation or reconstruction of such signals. In different applications, the estimation error that we want to minimize could be the mean squared error or other metrics such as the mean absolute error and the support set error. Seeing that minimizing different error metrics is useful in  MMV problems, we study information-theoretic performance limits for MMV signal estimation with arbitrary additive error metrics. We also propose a message passing algorithmic framework that achieves the optimal performance, and rigorously prove the optimality of our algorithm for a special case. We further conjecture the optimality of our algorithm for some general cases, and back it up through numerical examples. As an application of our MMV algorithm, we propose a novel setup for active user detection in multi-user communication and demonstrate the promise of our proposed setup.
\end{abstract}


\section{Introduction}
Many systems in science and engineering can be approximated by a linear model,
where a signal $\x \in \mathbb{R}^N$ is recorded via a measurement matrix $\A\in\mathbb{R}^{M \times N}$, and then contaminated by a measurement channel,
\begin{equation}\label{eq:SMV}
\w = \A \x,\ y_m = \mathcal{Z}(w_m), \forall m\in \{1,\ldots,M\},
\end{equation}
where $y_m,m\in\{1,\ldots,M\}$, are the entries of the measurements $\y\in\mathbb{R}^M$, and the measurement channel $\mathcal{Z}(\cdot)$ is characterized by a probability density function (pdf), $f(y_m|w_m)$.
The goal is to estimate $\x$ from the measurements $\y$ given knowledge of $\A$
and a model for the measurement channel $f(y_m|w_m),\forall m$. We call such a system the {\em single measurement vector (SMV)} problem. 

In many applications, the signal acquisition systems are  distributed, where $J$ measurement matrices measure $J$ different signals individually. The key difference between such a system and $J$ individual SMV's, is that these $J$ signals are somewhat dependent.
An example of a model containing such dependencies is the multi-measurement vector (MMV) problem~\cite{chen2006trs,cotter2005ssl,Mishali08rembo,Berg09jrmm,LeeKimBreslerYe2011,LeeBreslerJunge2012,YeKimBresler2015}. 
The MMV problem
considers the estimation of a set of dependent signals, and has applications such as  magnetic resonance imaging
with multiple coils~\cite{JuYeKi07,JuSuNaKiYe09}, active user detection in multi-user communication~\cite{FletcherRanganGoyal2009,Boljanovic2017}, and diffuse optical tomography using multiple
illumination patterns~\cite{LeeKimBreslerYe2011}. In MMV, thanks to the dependencies among different signals, the number of sparse coefficients that can be successfully estimated
increases with the number of
measurements. This property was evaluated rigorously for noiseless
measurements using
$l_0$ minimization~\cite{DuarteWakinBaronSarvothamBaraniuk2013}, if the underlying signals share the same sparse supports. A non-rigorous replica analysis of MMV with measurement noise also shows the benefit of having more signal vectors~\cite{ZhuBaronKrzakala2017IEEE,ZhuDissertation2017}.

{\bf Related work:}
There are many estimation approaches for MMV problems. These include greedy algorithms such as SOMP~\cite{tropp2006ass,chen2006trs},
$l_1$ convex relaxation~\cite{malioutov2005ssr,tropp2006ass2}, and M-FOCUSS~\cite{cotter2005ssl}. REduce MMV and BOost (ReMBo) has
been shown to outperform conventional methods~\cite{Mishali08rembo}, and subspace methods have also
been used to solve MMV problems~\cite{LeeBreslerJunge2012,YeKimBresler2015}. However, these algorithms cannot handle the case of $J$ {\em different} measurement matrices.
Statistical approaches~\cite{ZinielSchniter2011} often achieve the oracle minimum mean squared error (MMSE). However, when running estimation algorithms for MMV problems, one might be interested in minimizing some other error. For example, if estimating the underlying signal is important, one could use the mean squared error (MSE) metric; when there might be outliers in the estimated signal, using the mean absolute error (MAE) metric might be more appropriate. Seeing that there is no prior work discussing the optimal performance with user-defined error metrics, we study the optimal performance with user-defined additive error metrics in MMV problems
where the signals share  common sparse supports, and each entry of the measurements is contaminated by parallel measurement channels (i.e., the channel in~\eqref{eq:SMV} satisfies $f(\y|\w)=\prod_{m=1}^M f(y_m|w_m)$). Note that a specific error metric, the MSE, has been studied in Zhu et al.~\cite{ZhuBaronKrzakala2017IEEE}, 
which
focuses on the MSE performance limits (i.e., the MMSE) of MMV signal estimation. In contrast, this work explores performance limits and designs an algorithm that can minimize {\em arbitrary additive error metrics} beyond MSE.

{\bf Contributions:}
This paper combines insights from Zhu et al.~\cite{ZhuBaronKrzakala2017IEEE} and Tan and coauthors~\cite{Tan2012SSP,Tan2014}, thus yielding a stronger understanding of the MMV problem, which could be extended in future work to other distributed signal acquisition settings, beyond MMV. To be more specific, we make several contributions in this paper. First, by extending Tan and coauthors~\cite{Tan2012SSP,Tan2014} we provide an algorithm based on a message passing (MP) framework~\cite{RanganGAMP2011ISIT,Krzakala2012probabilistic} that can be adapted to minimize the expected error for arbitrary additive error metrics. 
Our algorithm first runs MP until it converges or reaches some stopping criteria, and we then
denoise MP's output using a denoiser that minimizes the given additive error metric. Second, we 
prove rigorously that our algorithm is optimal in a specific SMV case,\footnote{In SMV, our algorithm closely resembles the one proposed by Tan and coauthors~\cite{Tan2012SSP,Tan2014}, with the difference being Tan and coauthors rely on relaxed belief propagation (relaxed BP, an MP algorithm)~\cite{Rangan2010CISS} and do not provide rigorous proofs for the optimality of their algorithm.} and further conjecture the optimality of our algorithm in MMV.
Third, as an example, we derive performance limits for MAE and mean weighted support set error (MWSE) by designing the corresponding optimal denoisers, based on the scalar channel noise variance (Section~\ref{sec:MPA}) derived from replica analysis (Appendix~\ref{app:inverMMSE})~\cite{ZhuBaronKrzakala2017IEEE}. Simulation results show the superiority and optimality of our algorithm. We note in passing that having more signal vectors in MMV helps reduce the estimation error. Finally, as an application of MMV and our algorithm, we propose a novel setup for active user detection in multi-user communications (details in Section~\ref{sec:app}) and demonstrate the promise of our proposed setup through simulation.

{\bf Organization:} The remainder of the paper is organized as follows.
We introduce our problem setting and MP algorithms in Section~\ref{sec:background}. Our  algorithmic framework, which can minimize arbitrary additive error metrics, is proposed in Section~\ref{sec:achievable}; we rigorously prove the optimality of our algorithm for an SMV case and conjecture the optimality of our algorithm for  MMV cases in Section~\ref{sec:converse}. For some example error metrics, we derive the corresponding optimal algorithms, together with the theoretical limits for these errors, in Section~\ref{sec:example}. Synthetic simulation results are discussed in Section~\ref{sec:numeric_synth}, followed by an application of our metric-optimal algorithm to a real-world problem in Section~\ref{sec:app}. We conclude in Section~\ref{sec:conclusion}.

{\bf Notations:} In this paper, bold capital letters represent matrices, bold lower case letters represent vectors, and normal font letters represent scalars. The $m$-th entry (scalar) of a vector $\z$ is denoted by $z_m$.
\section{Problem Setting and Background}
\label{sec:background}
\subsection{Problem setting}
{\bf Signal model}:
We consider an ensemble of $J$ signal vectors, $\x^{(j)}\in\mathbb{R}^N,\ j\in\{1,\ldots,J\}$, where $j$ is the index of the signal.
Consider a {\em super-symbol} $\x_n=[x_n^{(1)},\ldots,x_n^{(J)}],\ n\in\{1,\ldots,N\}$; all super-symbols in this paper are row vectors. The super symbol $\x_n$ follows a $J$-dimensional distribution,
\begin{equation}\label{eq:jsm}
f(\x_n)=\rho \phi(\x_n)+(1-\rho)\delta(\x_n),
\end{equation}
where $\rho\in (0,1)$ determines the percentage of non-zeros in the signal and is called the sparsity rate, $\phi(\x_n)$ is a $J$-dimensional pdf, and 
   $\delta(\x_n)=\left\{
                \begin{array}{ll}
                 1, & \x_n = \textbf{0},\\
                 0, & \text{else}.
                \end{array}
      		\right.\\
$

\begin{myDef}[Joint sparsity with common supports]\label{def:jointly_sparse}
{\em Ensembles of signals are called jointly sparse signals with common sparse supports when they obey~\eqref{eq:jsm}.}
\end{myDef}

Note that there are other types of joint sparsity~\cite{BaronDCStech} that fit into the MMV framework. For example, an MMV model with signal vectors that have slowly changing supports is discussed in Ziniel and Schniter~\cite{ZinielSchniter2013MMV}. Since this paper only focuses on the MMV problem with signals sharing common sparse supports~\eqref{eq:jsm}, we refer to joint sparsity with common sparse supports as joint sparsity for brevity.

{\bf Measurement models}: 
Each signal $\x^{(j)}$ is measured by
a measurement matrix $\A^{(j)}\in\mathbb{R}^{M\times N}$ before being corrupted by a random measurement channel,
\begin{equation}\label{eq:linearMixing}
\begin{split}
\w^{(j)}&=\A^{(j)}\x^{(j)},\ y_m^{(j)}=\mathcal{Z}(w_m^{(j)}),\\
 &\quad m\in\{1,\ldots,M\},\ j\in\{1,\ldots,J\},
\end{split}
\end{equation}
where $y_m^{(j)}, m\in \{1,\ldots,M\}$ are the entries of the measurements $\y^{(j)}$, and the measurement channel $\mathcal{Z}(\cdot)$ is characterized by the pdf $f(y_m^{(j)}|w_m^{(j)})$.
In this paper, we only focus on independent and identically distributed   (i.i.d.) parallel measurement channels, i.e., the pdf's $f(y_m^{(j)}|w_m^{(j)}), \forall m,j$, are identical and there is no cross-talk among different channels; our proposed algorithm is readily extended to parallel channels with different $f(y_m^{(j)}|w_m^{(j)}), \forall m,j$.
When the number of signal vectors $J=1$, we call this MMV model \eqref{eq:linearMixing} an SMV problem~\eqref{eq:SMV}.

\begin{myDef}[Large system limit~\cite{GuoWang2008}]\label{def:largeSystemLimit}
{\em The signal length $N$ scales to infinity, and the
number of measurements $M=M(N)$ depends on $N$ and also scales to infinity, where
the ratio approaches a positive constant $R$,}
\begin{equation*}
\lim_{N\rightarrow\infty} \frac{M(N)}{N} = R>0.
\end{equation*}
\end{myDef}
We call $R$ the measurement rate.

For MMV problems, we are given the matrices $\A^{(j)}$ and measurements $\y^{(j)},\ \forall j$, as well as knowledge of the measurement channel~\eqref{eq:linearMixing}. Our task is to estimate the underlying signal vectors $\x^{(j)},\ \forall j$. Suppose that the estimate is $\widehat{\x}^{(j)}$. 
Define $\X=[\x^{(1)},\cdots,\x^{(J)}]$ and $\widehat{\X}=[\widehat{\x}^{(1)},\cdots,\widehat{\x}^{(J)}]$. Therefore, $\X=[\x_1^T,\cdots,\x_N^T]^T$, where $\{\cdot\}^T$ denotes transpose. The estimation quality is quantified by a user-defined error metric $D_{\text{UD}}(\X,\widehat{\X})$, where the subscript UD denotes ``user-defined." We define this {\em additive error metric $D_{\text{UD}}(\cdot,\cdot)$ as}
\begin{equation*}
D_{\text{UD}}(\X, \widehat{\X})=\sij{n=1}{N}d_{\text{UD}}(\x_n,\widehat{\x}_n),
\end{equation*}
where $d_{\text{UD}}(\cdot,\cdot): \mathbb{R}^J\times \mathbb{R}^J \rightarrow \mathbb{R}$ is an arbitrary user-defined error metric on each super-symbol.
The smaller the $D_{\text{UD}}(\cdot,\cdot)$ is, the better the estimation is.

\subsection{Message passing algorithms}\label{sec:MPA}
Message passing (MP) algorithms consider a factor graph~\cite{RanganGAMP2011ISIT,Krzakala2012probabilistic,BarbierKrzakala2017IT}, which expresses the relation between the signals $\x$ and measurements $\y$. We begin by discussing the factor graph for SMV, followed by that of MMV.

\begin{figure}[t]
\centering
\includegraphics[width=8cm]{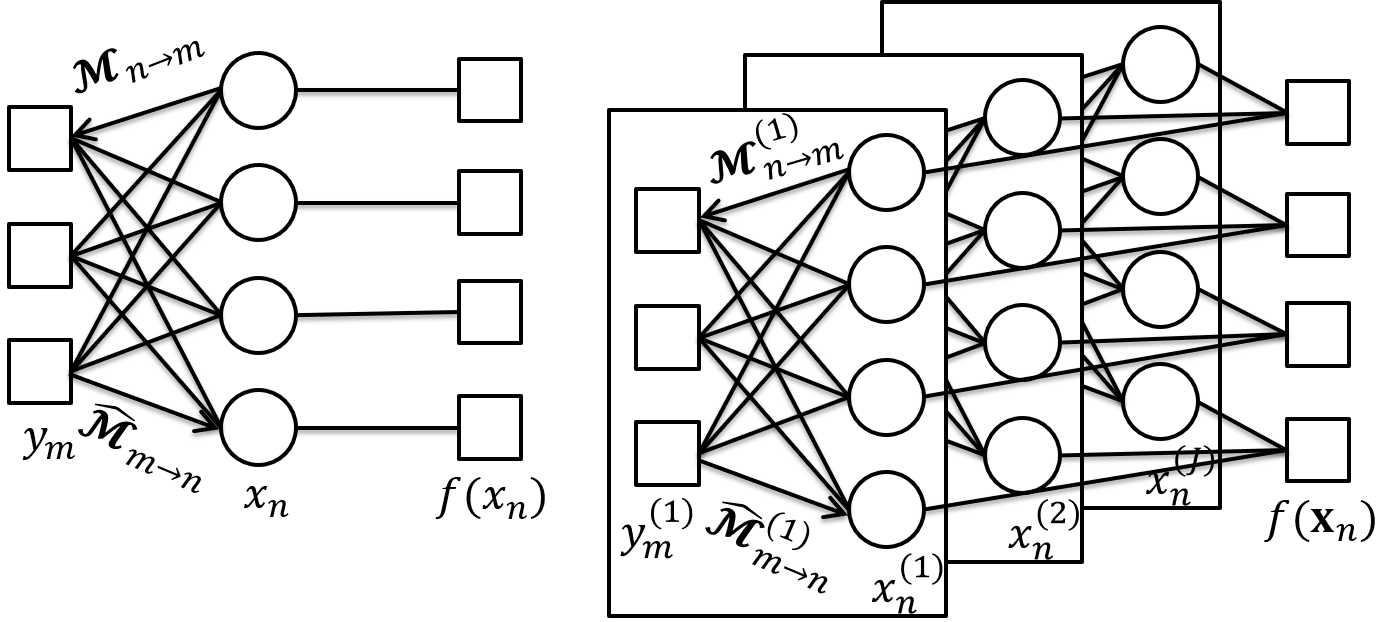}
\caption{Factor graph for SMV (left) and MMV (right).}
\label{fig:factorGraph}
\end{figure}

{\bf Factor graph for SMV:}
The left panel of Fig.~\ref{fig:factorGraph} illustrates the factor graph concept for an SMV problem~\eqref{eq:SMV} with i.i.d. entries in the signal $\x$. The round circles are the variable nodes (representing the distribution of the signal), and the squares denote the factor nodes (representing the measurement channel). The variables $x_n,\ \forall n$, are driven by each factor node $f(x_n)$ individually, because the signal has i.i.d. entries $x_n$. 
There are two types of messages passed in the factor graph shown on the left panel of Fig.~\ref{fig:factorGraph}: the 
message passed by the variable node $x_n$ to the factor node $y_m$, $\mathcal{M}_{n\rightarrow m}(x_n)$, and the message passed by the factor node $y_m$ to the variable node $x_n$, $\widehat{\mathcal{M}}_{m\rightarrow n}(x_n)$. According to the literature on MP algorithms~\cite{RanganGAMP2011ISIT,Krzakala2012probabilistic,BarbierKrzakala2017IT}, we have the following relation:
\begin{equation*}
\begin{split}
\mathcal{M}_{n\rightarrow m}(x_n) &=\frac{1}{Z_{n\rightarrow m}}f(x_n) \prod_{\widehat{m} \neq m} \widehat{\mathcal{M}}_{\widehat{m}\rightarrow n}(x_n), \\
\widehat{\mathcal{M}}_{m\rightarrow n}(x_n) &=\frac{1}{Z_{m\rightarrow n}} \int f(y_m | \x) \prod_{\widehat{n}\neq n} \mathcal{M}_{\widehat{n}\rightarrow m}(x_{\widehat{n}})\prod_{\widehat{n}\neq n}dx_{\widehat{n}},
\end{split}
\end{equation*}
where we use a single integral sign to denote a multi-dimensional integration for brevity, and $Z_{n\rightarrow m}$ and $Z_{m\rightarrow n}$ are normalization factors.
The aim of this paper is not the detailed derivation of MP algorithms. Instead, the key property of MP utilized in this paper is that MP converts~\eqref{eq:SMV} into the following equivalent scalar channel,
\begin{equation*}
\q=\x+\v,
\end{equation*}
where $\q$ is the noisy {\em pseudo data}, and $\v$ is the equivalent scalar channel additive white Gaussian noise (AWGN) whose variance $\Delta_v$ can be approximated. After obtaining $\q$ and $\Delta_v$, each variable node $x_n$ updates the estimate $\widehat{x}_n$ by denoising $q_n$.
If the signal entries are i.i.d. and the matrix is either i.i.d. or sparse and locally tree-like, then MP algorithms yield a density function $f(x_n|q_n)$ that is statistically equivalent to $f(x_n|\y)$~\cite{RanganGAMP2011ISIT}.

{\bf Factor graph for MMV:}
An MMV problem with jointly sparse signals can be expressed as the factor graph shown in the right panel of Fig.~\ref{fig:factorGraph}. We can see that the MP in each channel is similar to an SMV problem. The only difference is that the $J$ variable nodes $x_n^{(j)},\ j\in\{1,\ldots,J\}$, for fixed $n$ are driven by one factor node $f(\x_n)$. By grouping the entries from different signal vectors together into super-symbols as in~\eqref{eq:jsm},  we have i.i.d. super-symbols. 
The noisy super-symbol pseudo data $\q_n=[q_n^{(1)},\ldots,q_n^{(J)}]$ is denoised, in order to update the estimate for $\x_n=[x_n^{(1)},\ldots,x_n^{(J)}]$. 

\section{Main Results}\label{sec:main}
We first present our metric-optimal algorithm in Section~\ref{sec:achievable} and then in Section~\ref{sec:converse} we rigorously prove that our metric-optimal algorithm is optimal in the SMV case under certain conditions.\footnote{Our proof is based on the approximate message passing algorithm~\cite{DMM2009}, while Tan and coauthors~\cite{Tan2012SSP,Tan2014} rely on relaxed BP and do not provide rigorous proofs.} Based on our proof for the SMV case, we conjecture
that the proposed algorithm is optimal for arbitrary additive error metrics in the MMV case.
\subsection{Achievable part: Metric-optimal estimation algorithm}\label{sec:achievable}
The metric-optimal algorithm consists of two parts, as illustrated in Algorithm~\ref{algo:metric_opt_MMV}. We first run an MP algorithm (Algorithm~\ref{algo:AMP_MMV} provides an implementation of an MP algorithm) to get the noisy pseudo data $\q^{(j)}, \forall j$, and the noise variance $\Delta_v$ (details below). Next, we  denoise $\q^{(j)}, \forall j$, using an optimal denoiser tailored to minimize the given error metric. The following discusses both parts in detail.

\begin{algorithm}[t]
\caption{Metric-optimal algorithm for MMV}
\label{algo:metric_opt_MMV}
\begin{algorithmic}[1]
\\{\bf Inputs:} Measurements $\y^{(j)}$ and matrices $\A^{(j)}, \forall j$
\\{\bf Part 1 (Algorithm~\ref{algo:AMP_MMV}):} Obtain pseudo data $\q^{(j)},\ \forall j$, and scalar channel noise variance $\Delta_v$ from MP($\y^{(j)}, \A^{(j)}, \forall j$)
\\{\bf Part 2 (examples in Section~\ref{sec:example}):} 
Obtain optimal estimate $\widetilde{\x}^{(j)}$ from denoiser using $\q^{(j)},\Delta_v, \forall j$
\\{\bf Outputs:} $\widetilde{\x}^{(j)},\ \forall n$
\end{algorithmic}
\end{algorithm}

\begin{algorithm}[t]
\caption{GAMP for MMV}
\label{algo:AMP_MMV}
\begin{algorithmic}[1]
\\{\bf Inputs:} Maximum number of iterations $t_{\text{max}}$, threshold $\epsilon$, sparsity rate $\rho$, noise variance $\Delta_z$, measurements $\y^{(j)}$, and measurement matrices $\A^{(j)}, \forall j$
\\{\bf Initialize:} $t=1,\delta=\infty,\k^{(j)}=\y^{(j)},\Theta_m^{(j)}=0,s^{(j)}_n=\rho\Delta_z,\widehat{x}_n^{(j)}=0, h_m^{(j)}=0, \forall m,n,j$
\While{$t<t_{\text{max}}$ and $\delta>\epsilon$}
\For{$j\leftarrow 1$ to $J$} \label{line:beginForLoop}
\\\quad\quad\quad$\boldsymbol{\Theta}^{(j)}=(\A^{(j)})^2 \s^{(j)}$\label{line:Theta}
\\\quad\quad\quad$\k^{(j)}=\A^{(j)} \widehat{\x}^{(j)}-\text{diag}(\boldsymbol{\Theta}^{(j)}) \h^{(j)}$
\For{$m\leftarrow 1$ to $M$}
\\\quad\quad\quad\quad$h_m^{(j)}=g_{\text{out}}\l(k_m^{(j)},y_m^{(j)},\Theta_m^{(j)}\r)$\label{line:g_out} 
\\\quad\quad\quad\quad $r_m^{(j)} = -\frac{\partial}{\partial k_m^{(j)}} g_{\text{out}}\l(k_m^{(j)},y_m^{(j)},\Theta_m^{(j)}\r)$\label{line:deriv_g_out}
\EndFor
\\\quad\quad\quad // Scalar channel noise variance
\\\quad\quad\quad$\Delta_v^{(j)}=\l\{\frac{1}{N} \mathbf{1}^T \l[(\A^{(j)})^T\r]^2 \mathbf{r}^{(j)}\r\}^{-1}$ \label{line:scalarNoiseVar}
\\\quad\quad\quad$\q^{(j)}=\widehat{\x}^{(j)}+\Delta_v^{(j)} (\A^{(j)})^T \h^{(j)}$ // Pseudo data
\\\quad\quad\quad$\widehat{\a}^{(j)}=\widehat{\x}^{(j)}$ // Save current estimate
\EndFor \label{line:endForLoop}
\\\quad\ $\Delta_v=\sum_{j=1}^J \Delta_v^{(j)}$\label{line:mean_delta}
\For{$n\leftarrow 1$ to $N$}
\\\quad\quad\quad$\widehat{\x}_n=f_{a_n}(\Delta_v,\q_n)$ // Estimate\label{line:mean}
\\\quad\quad\quad$\s_n=[s_n^{(1)},\ldots,s_n^{(J)}]=f_{v_n}(\Delta_v,\q_n)$ // Variance\label{line:var}
\EndFor
\\\quad\ \ $t=t+1$ // Increment iteration index
\\\quad\ \ $\delta=\frac{1}{NJ}\sum_{n=1}^N\sum_{j=1}^J\l(\widehat{x}^{(j)}_n-\widehat{a}^{(j)}_n\r)^2$ // Change
\EndWhile
\\{\bf Outputs:} Estimate $\widehat{\x}^{(j)}$, pseudo data $\q^{(j)},\ \forall j$, and scalar channel noise variance $\Delta_v$
\end{algorithmic}
\end{algorithm}

{\bf MP algorithm:}
For the first part, we modify the generalized approximate message passing (GAMP) algorithm~\cite{RanganGAMP2011ISIT}, which is an implementation of MP, and list the pseudo code in Algorithm~\ref{algo:AMP_MMV}. The notation diag$(\boldsymbol{\Theta}^{(j)})$ 
denotes a diagonal matrix whose entries along the diagonal are $\boldsymbol{\Theta^{(j)}}$,
and the power-of-two in Lines~\ref{line:Theta} and~\ref{line:scalarNoiseVar} is applied element-wise. The function $g_{\text{out}}$ in Lines~\ref{line:g_out} and~\ref{line:deriv_g_out} is given by
\begin{equation}\label{eq:g_out}
g_{\text{out}}\l(k,y,\Theta\r)=\frac{1}{\Theta}(\mathbb{E}[w|k,y,\Theta]-k),
\end{equation}
where we omit the subscripts and super-scripts for brevity, and the expectation is taken over the pdf,
\begin{equation}\label{eq:prob_w}
f(w|k,y,\Theta)\propto f(y|w)\text{exp}\l[-\frac{(w-k)^2}{2\Theta}\r].
\end{equation}

For the special case of  AWGN channels,
\begin{equation}\label{eq:AWGN}
y=w+z,
\end{equation} 
where $z\sim \mathcal{N}(0,\Delta_z)$, we obtain
$g_{\text{out}}(k,y,\Theta)=\frac{y-k}{\Delta_z+\Theta}$~\cite{RanganGAMP2011ISIT}. In Appendix~\ref{app:logit}, we also briefly present the derivation for an i.i.d. parallel logistic channel,
\begin{equation}\label{eq:logit}
f(y|w)=\delta(y-1)\frac{1}{1+\text{exp}(-aw)}+\delta(y)\frac{\text{exp}(-aw)}{1+\text{exp}(-aw)},
\end{equation}
where $a$ is a scaling factor.\footnote{Byrne and Schniter~\cite{ByrneSchniter2015ArXiv} describe without detail how to derive  $g_{\text{out}}(\cdot)$~\eqref{eq:g_out} for i.i.d. parallel logistic channels~\eqref{eq:logit}; we present a detailed derivation for completeness in Appendix~\ref{app:logit}, and do not claim it as a contribution.} 

For the special case of i.i.d. {\em joint Bernoulli-Gaussian signals} where $\phi(\x_n)\sim \mathcal{N}(0,\mathbb{I})$ in~\eqref{eq:jsm} and $\mathbb{I}$ is an identity matrix, $f_{a_n}$ and $f_{v_n}$ in Lines~\ref{line:mean} and~\ref{line:var} are given below,
\begin{equation}\label{eq:denoiser}
f_{a_n}(\Delta_v,\q_n)=\frac{\rho}{C(\Delta_v+1)}\q_n,
\end{equation}
\begin{equation*}
f_{v_n}(\Delta_v,\q_n)\!=\!-[f_{a_n}(\Delta_v,\q_n)]^2+\frac{\rho}{C(\Delta_v+1)}\!\l[\frac{\q_n^2}{\Delta_v+1}\!+\!\Delta_v\r],
\end{equation*}
where $\q_n^2=\l[\l(q_n^{(1)}\r)^2,\ldots,\l(q_n^{(J)}\r)^2\r]$ and
\begin{equation*}
C=\rho+(1-\rho)\l(1+\frac{1}{\Delta_v}\r)^{\frac{J}{2}}\exp\l[-\frac{\q_n\q_n^T}{2\Delta_v(\Delta_v+1)}\r].
\end{equation*}
Notice that in Line~\ref{line:scalarNoiseVar} we take the mean of a vector to obtain a scalar $\Delta_v^{(j)}$, which is the average  of the variances for the estimates of signal entries $x_n^{(j)}$. This is because the super-symbols $\x_n, \forall n$, of the signals are i.i.d and the $J$ measurement channels are i.i.d. For the same reason, $\Delta_v^{(j)},\ j\in\{1,\ldots,J\}$, should be close to each other; hence, Line~\ref{line:mean_delta}. Note that Algorithm~\ref{algo:AMP_MMV} assumes that the entries of $\A^{(j)}$ scale with $\frac{1}{\sqrt{N}}$, and is a more generic form of an algorithm from our prior work with Krzakala~\cite{ZhuBaronKrzakala2017IEEE}.

{\bf Metric-optimal denoiser:}
The second part of our metric-optimal algorithm takes as inputs the noisy pseudo data $\q_n=\x_n+\v_n$ and the estimated variance of $\v_n,\ \Delta_v$, from the MP algorithm. Using Bayes' rule, we can derive the posterior $f(\x_n|\q_n)$ and use $f(\x_n|\q_n)$ to formulate the optimal estimator in the sense of the user-defined error metric.
The optimal estimate is
\begin{equation}\label{eq:estimator}
\widetilde{\x}_n=\arg\min_{\widehat{\x}_n} \int d_{\text{UD}}(\x_n,\widehat{\x}_n) f(\x_n|\q_n) d\x_n.
\end{equation}

\subsection{Converse part: The optimal estimate}\label{sec:converse}
The reason why~\eqref{eq:estimator} is optimal is based on the insight from Rangan~\cite{RanganGAMP2011ISIT}  that in SMV the density function $f(x_n|q_n)$  converges to the posterior $f(x_n|\y)$. A rigorous proof for a certain SMV case is provided below, followed by our conjecture that~\eqref{eq:estimator} is optimal in the MMV case.

\begin{myLemma}[Optimality of Algorithm~\ref{algo:metric_opt_MMV} in SMV]\label{lemma:optSMV}
Consider an SMV~\eqref{eq:SMV} in the large system limit (Definition~\ref{def:largeSystemLimit}) with an AWGN measurement channel, $f(y_m|w_m)=\frac{1}{\sqrt{2\pi\sigma_Z^2}}\exp\l[\frac{(y_m-w_m)^2}{2\sigma_Z^2}\r], \forall m$. The estimate  $\widetilde{\x}_n$~\eqref{eq:estimator} is optimal in the sense that
\begin{equation}\label{eq:optSMV}
\lim_{N\rightarrow\infty} \frac{1}{N}\sum_{n=1}^N d_{\text{UD}}(\widetilde{x}_n,x_n) = \text{MUDE},
\end{equation} 
where MUDE denotes the minimum user-defined error, if all the conditions below hold.
\begin{enumerate}
\item The entries of the measurement matrix are i.i.d. Gaussian, $A_{mn}\sim \mathcal{N}(0,\frac{1}{N})$,
\item the signal entries are i.i.d. with bounded fourth moment $\mathbb{E}[X^4]<B$, where $B$ is some constant,
\item  the free energy given by replica analysis has one fixed point~\cite{ZhuBaronCISS2013,ZhuBaronKrzakala2017IEEE,Krzakala2012probabilistic},\footnote{Free energy is a term brought from statistical physics and is used to describe the interaction between the signals and the measurement matrices in linear models~\cite{ZhuBaronCISS2013,ZhuBaronKrzakala2017IEEE,Krzakala2012probabilistic}. When the free energy has two fixed points, (G)AMP is not optimal with i.i.d. Gaussian matrices, and neither is Algorithm~\ref{algo:metric_opt_MMV}. We refer interested readers to the literature for detailed discussions~\cite{ZhuBaronCISS2013,ZhuBaronKrzakala2017IEEE,Krzakala2012probabilistic}.} 
\item the user-defined error metric $d_{\text{UD}}(\cdot,\cdot)$ is pseudo-Lipschitz~\cite{Bayati2011},\footnote{Pseudo-Lipschitz is a concept discussed in Bayati and Montanari~\cite{Bayati2011}: For $k\geq 1$, we say a function $\phi: \mathbb{R}^m \rightarrow \mathbb{R}$ is
{\em pseudo-Lipschitz} of order $k$ if there exists a constant $L>0$ such that for all $\x,\y\in \mathbb{R}^m$: $|\phi(\x)-\phi(\y)| \leq L(1+\|\x\|^{k-1}+\|\y\|^{k-1})\|\x-\y\|$.}
\item the optimal estimator~\eqref{eq:estimator} as a function of $q_n,\ \widetilde{x}_n(q_n): \mathbb{R}\rightarrow \mathbb{R}$
, is Lipschitz continuous, and
\item Part~1 converges before entering Part~2 in Algorithm~\ref{algo:metric_opt_MMV}.
\end{enumerate}
\end{myLemma}

\begin{proof}
According to Theorem~1 in Bayati and Montanari~\cite{Bayati2011}, we know that 
$\lim_{N\rightarrow\infty} \frac{1}{N}\sum_{n=1}^N d_{\text{UD}}(\widetilde{x}_n(q_n),x_n) = \mathbb{E}\l[d_{\text{UD}}(\widetilde{x}_{n}(X+\Delta_v Z),X)\r]$, where $X$ is a random variable following the same distribution as the signal entries, $Z\sim \mathcal{N}(0,1)$, and $\Delta_v$ is given in Line~\ref{line:mean_delta} of Algorithm~\ref{algo:AMP_MMV}. The question boils down to showing that $\text{MUDE} = \mathbb{E}\l[d_{\text{UD}}(\widetilde{x}_{n}(X+\Delta_v Z),X)\r]$.
In fact, when Algorithm~\ref{algo:AMP_MMV} converges, $\Delta_v$ corresponds to the scalar channel noise variance associated with the MMSE given by replica analysis~\cite{Bayati2011,RushVenkataramanan2016}. In addition, the MMSE provided by replica analysis is proved  to be exact under the conditions asserted in Lemma~\ref{lemma:optSMV}~\cite{ReevesPfister2016}. Therefore, $\mathbb{E}\l[d_{\text{UD}}(\widetilde{x}_{n}(X+\Delta_v Z),X)\r]$ is indeed the MUDE. Hence, we proved~\eqref{eq:optSMV}, and the estimator~\eqref{eq:estimator} is optimal.
\end{proof}

\textbf{Remark 1:} The optimality of the metric-optimal estimator $\widetilde{\x}_n$~\eqref{eq:estimator} is stated in the sense that the ensemble mean user-defined error converges almost surely to the MUDE. It is possible that there exist other estimators that achieve the  MUDE. 

\textbf{Remark 2:} The proof is made possible by linking three rigorous proofs~\cite{ReevesPfister2016,Bayati2011,RushVenkataramanan2016} from the prior art. That said, numerical examples in Section~\ref{sec:numeric_synth} demonstrate that Algorithm~\ref{algo:metric_opt_MMV} yields promising results even if the conditions required by Lemma~\ref{lemma:optSMV} are not met.

After proving the optimality of our metric-optimal algorithm in the SMV scenario with certain conditions, we state the following conjecture.
\begin{myConj}\label{conj:2}
In the large system limit, for the MMV model \eqref{eq:linearMixing} with the signal in~\eqref{eq:jsm} and the user-defined additive error metric $d_{\text{UD}}(\x_n,\widehat{\x}_n)$, the optimal estimate of the signal vectors is
\begin{equation*}
\widetilde{\x}_n=\arg\min_{\widehat{\x}_n} \mathbb{E}[d_{\text{UD}}(\x_n,\widehat{\x}_n)|\q_n].
\end{equation*}
\end{myConj}

As the reader can see from the proof of Lemma~\ref{lemma:optSMV}, in order to rigorously prove Conjecture~\ref{conj:2}, we need to show ({\em i}) the MMSE given by the replica analysis for the MMV case~\cite{ZhuBaronKrzakala2017IEEE} is exact, ({\em ii}) the scalar channel noise variance $\Delta_v$ in Algorithm~\ref{algo:AMP_MMV} corresponds to the MMSE given by replica analysis, and ({\em iii}) a result similar to Theorem 1 in Bayati and Montanari~\cite{Bayati2011} holds. None of these three results exists in the prior art, so we do not foresee what exact conditions are needed for Conjecture~\ref{conj:2}. Proving these three results (and hence Conjecture~\ref{conj:2}) is beyond the scope of this work. Instead, we provide some intuition that explains why we believe Algorithm~\ref{algo:metric_opt_MMV} is optimal in the MMV scenario.

In the SMV case, $f(x_n|q_n)$ converges to the posterior $f(x_n|\y)$ in relaxed BP~\cite{RanganGAMP2011ISIT}. As an extension to the MMV case, we intuitively think that $f(\x_n|\q_n)$ would converge to the posterior $f(\x_n|\{\y^{(j)}\}_{j=1}^J)$,
based on two observations: ({\em i}) the measurement channels are i.i.d., so that it suffices to update the estimate of the channel by passing the messages $\mathcal{M}_{n\rightarrow m}(x_n^{(j)})$ and $\widehat{\mathcal{M}}_{m\rightarrow n}(x_n^{(j)})$ for different $j$ individually, and ({\em ii}) the super-symbols $\x_n$ are i.i.d., so that denoising each super-symbol individually accounts for all the information needed to update the estimate. 

The optimal estimate of each super-symbol $\x_n$ in the signal vectors is
\begin{equation}\label{eq:estimatorTrue}
\widetilde{\x}_{\text{true},n}=\arg\min_{\widehat{\x}_n} \int d_{\text{UD}}(\x_n,\widehat{\x}_n) f(\x_n|\{\y^{(j)}\}_{j=1}^J) d\x_n.
\end{equation}
Comparing~\eqref{eq:estimator} and~\eqref{eq:estimatorTrue}, provided that 
$f(\x_n|\q_n)$ converges to the posterior $f(\x_n|\{\y^{(j)}\}_{j=1}^J)$,
we have 
\begin{equation*}
\widetilde{\x}_n=\arg\min_{\widehat{\x}_n} \mathbb{E}[d_{\text{UD}}(\x_n,\widehat{\x}_n)|\q_n]\approx \widetilde{\x}_{\text{true},n},
\end{equation*}
which results in Conjecture~\ref{conj:2}.

\section{Example Metric-optimal Estimators and Performance Limits}\label{sec:example}\
In order to derive metric-optimal estimators, we need to know the scalar channel noise variance, $\Delta_v$. 
Below we obtain the optimal estimator, which is then used as part of Algorithm~\ref{algo:metric_opt_MMV} and to evaluate the performance limits of our metric-optimal algorithm.

For i.i.d. matrices and AWGN channels~\eqref{eq:AWGN}, replica analysis in our previous work with Krzakala~\cite{ZhuBaronKrzakala2017IEEE} yields the information-theoretic scalar channel noise variance $\Delta_v$ for message passing algorithms, which characterizes the posterior $f(\x_n|\q_n)$.\footnote{Our work with Krzakala~\cite{ZhuBaronKrzakala2017IEEE} focuses on a diagonal covariance matrix for $\x_n$ in~\eqref{eq:jsm}. A recent work by Hannak et al.~\cite{Hannak2017} extends our work~\cite{ZhuBaronKrzakala2017IEEE} to non-diagonal covariance matrices for $\x_n$. Following Hannak et al.~\cite{Hannak2017}, we can extend the performance limits analysis in this paper to non-diagonal covariance matrices for $\x_n$.} Hence, we can obtain the information-theoretic optimal performance with arbitrary additive error metrics. For other types of matrices,
our replica analysis~\cite{ZhuBaronKrzakala2017IEEE} does not hold. Nevertheless, MP algorithms still yield a posterior $f(\x_n|\q_n)$, which we conjecture converges to the true posterior $f(\x_n|\{\y^{(j)}\}_{j=1}^J)$ (Conjecture~\ref{conj:2}). Hence, we can still assume that the scalar channel noise variance $\Delta_v$ is known. Based on the known variance $\Delta_v$, we build metric-optimal estimators~\eqref{eq:estimator} for two examples, mean weighted support set error (Section~\ref{sec:weightedSupportSet}) and mean absolute error (Section~\ref{sec:MAE}).

\subsection{Mean weighted support set error (MWSE)}\label{sec:weightedSupportSet}
{\bf MWSE-optimal estimator:}
In support set estimation, the goal is to estimate the support of the signal, which is 1 if the corresponding entry in the signal is non-zero and 0 if it is zero. There are
two types of errors in support set estimation: false alarms (support is 0, but estimated to be 1) and misses (support is 1, but estimated to be 0). In some applications such as  medical imaging and radar detection, a miss may mean that the doctor misses an  illness of the patient, or the radar misses an incoming missile. Hence, the cost paid for a miss could be tremendous compared to a false alarm. There are other applications where a false alarm is more costly than a miss. For example, in  court, if an innocent person is mistakenly judged guilty, he/she will likely suffer a great deal. Therefore, we should weight these two errors differently in different applications. Let $b_n$ and $\widehat{b}_n$ be the true support and the estimated support of the $n$-th entry of the signal, respectively, and $\beta\in[0,1]$ is 
an application-dependent weight, which reflects the trade-off between the false alarms and misses. Hence, the MWSE given the pseudo data $\q_n$ is 
\begin{equation}\label{eq:MWSE}
\begin{split}
    &\text{MWSE}|\q_n=\mathbb{E}[d_{\text{WSE}}(b_n,\widehat{b}_n)|\q_n]=\\
    &\left\{
                \begin{array}{ll}
                 (1-\beta)\Pr(b_n=1|\q_n),&\   \widehat{b}_n=0\ \text{and}\ b_n=1,\\
                 \beta\Pr(b_n=0|\q_n),&\   \widehat{b}_n=1\ \text{and}\ b_n=0,\\
                 0,\quad\quad\quad&\ \widehat{b}_n=b_n,
                \end{array}
      		\right.\\
\end{split}
\end{equation}
where $\Pr(\cdot)$ denotes probability.
The optimal estimate $\widetilde{b}_n$ minimizes $\mathbb{E}[d_{\text{WSE}}(b_n,\widehat{b}_n)|\q_n]$~\eqref{eq:MWSE}, which implies
\begin{equation}\label{eq:weighted_optB}
    \widetilde{b}_n\!=\!\left\{
                \begin{array}{ll}
                 \!0,&\!   (1-\beta)\Pr(b_n=1|\q_n)\!\leq\! \beta\Pr(b_n=0|\q_n),\\
                 \!1,&\! (1-\beta)\Pr(b_n=1|\q_n)\!>\! \beta \Pr(b_n=0|\q_n).
                \end{array}
      		\right.\\
\end{equation}

Since $f(\q_n|b_n=0)=(2\pi\Delta_v)^{-\frac{J}{2}}\exp\l(-\frac{\q_n\q_n^T}{2\Delta_v}\r)$ and $f(\q_n|b_n=1)=[2\pi(\Delta_v+1)]^{-\frac{J}{2}}\exp\l[-\frac{\q_n\q_n^T}{2(\Delta_v+1)}\r]$, we have 
\begin{equation*}
\Pr(b_n=1|\q_n)=\frac{\frac{\rho}{[2\pi (1+\Delta_v)]^{\frac{J}{2}}}\e^{-\frac{\q_n\q_n^T}{2(1+\Delta_v)}}}{\frac{\rho}{[2\pi (1+\Delta_v)]^{\frac{J}{2}}}\e^{-\frac{\q_n\q_n^T}{2(1+\Delta_v)}}+\frac{(1-\rho)}{(2\pi \Delta_v)^{\frac{J}{2}}}\e^{-\frac{\q_n\q_n^T}{2\Delta_v}}},
\end{equation*}
and $\Pr(b_n=0|\q_n)=1-\Pr(b_n=1|\q_n)$.
Plugging $\Pr(b_n=1|\q_n)$ and $\Pr(b_n=0|\q_n)$ into~\eqref{eq:weighted_optB}, we have
\begin{equation}\label{eq:b_opt}
    \widetilde{b}_n=\left\{
                \begin{array}{ll}
                 0,&\quad   \q_n\q_n^T\leq\tau,\\
                 1,&\quad   \q_n\q_n^T>\tau,
                \end{array}
      		\right.\\
\end{equation}
where $\tau=2\Delta_v(1+\Delta_v)\ln\l[\frac{\beta(1-\rho)}{(1-\beta)\rho}\l(\frac{1+\Delta_v}{\Delta_v}\r)^{\frac{J}{2}}\r]$, and we remind the reader that $\rho$ is the sparsity rate.

{\bf Performance limits:}
Utilizing~\eqref{eq:b_opt} and taking expectation over the pseudo data $\q_n$ for $\mathbb{E}[d_{\text{WSE}}(b_n,\widehat{b}_n)|\q_n]$~\eqref{eq:MWSE}, we obtain the minimum MWSE (MMWSE),
\begin{equation}\label{eq:MMWSE_1}
\begin{split}
&\text{MMWSE}=\mathbb{E}[d_{\text{WSE}}(b_n,\widetilde{b}_n)]\\
&\quad =\int_{\q_n\q_n^T>\tau} \beta\Pr(b_n=0|\q_n)f(\q_n)d\q_n +\\
&\quad \quad \int_{\q_n\q_n^T\leq\tau} (1-\beta)\Pr(b_n=1|\q_n)f(\q_n)d\q_n.
\end{split}
\end{equation}
We have two integrals to simplify in~\eqref{eq:MMWSE_1}, where we show the first below, and the second can be obtained similarly. To derive the first integral, note that
\begin{equation}\label{eq:oneIntegral}
\small \int_{\q_n\q_n^T>\tau} \Pr(b_n\!=\!0|\q_n)f(\q_n)d\q_n\!=\!\Pr(\q_n\q_n^T\!>\!\tau, b_n\!=\!0).
\end{equation}
Next, we calculate the pdf of the random variable (RV) $g_n=\frac{\q_n\q_n^T}{\Delta_v}$ given $b_n=0$. Because the entries of $\q_n$ are i.i.d. $\mathcal{N}(0,\Delta_v)$ given $b_n=0$, $g_n$ follows the Chi-square distribution,
$f_G(g_n)=\frac{g_n^{\frac{J}{2}-1}\exp(-\frac{g_n}{2})}{2^{\frac{J}{2}}\Gamma(\frac{J}{2})}$,
where $\Gamma(\cdot)$ is the Gamma function.
Let $r_n=\Delta_v g_n=\q_n\q_n^T$, then we obtain
\begin{equation*}
f(r_n)=\frac{1}{\Delta_v}f_G(\frac{r_n}{\Delta_v})=\frac{r_n^{\frac{J}{2}-1}\exp\l[-\frac{r_n}{2\Delta_v}\r]}{(2\Delta_v)^{J/2}\Gamma(J/2)},
\end{equation*}
which helps to simplify~\eqref{eq:oneIntegral}. Therefore,~\eqref{eq:MMWSE_1} can be simplified,
\begin{equation}\label{eq:MMWSE_final}
\begin{split}
& \text{MMWSE}=\beta(1-\rho)\underbrace{\int_{r_n=\tau}^{\infty}\frac{r_n^{\frac{J}{2}-1}\exp\l[-\frac{r_n}{2\Delta_v}\r]}{(2\Delta_v)^{J/2}\Gamma(J/2)}dr_n}_{\Pr(\text{false alarm})}\\
& \quad +(1-\beta)\rho\underbrace{\int_{r_n=0}^{\tau}\frac{r_n^{\frac{J}{2}-1}\exp\l[-\frac{r_n}{2(1+\Delta_v)}\r]}{[2(1+\Delta_v)]^{J/2}\Gamma(J/2)}dr_n}_{\Pr(\text{miss})}.
\end{split}
\end{equation}

{\bf Hamming distance}: In digital wireless communication systems, the signal only takes discrete values. A useful error metric is the (per-entry) Hamming distance~\cite{Cover06}, which equals 1 if the estimate of an entry of the signal differs from the true value. (Section~\ref{sec:app} will present an example in wireless communication that minimizes the Hamming distance.)
The reader can verify that Hamming distance can be interpreted as a special case of  weighted support set error, where $\beta=0.5$ provides equal weight to both errors~\eqref{eq:MWSE}. That said, we provide more insights about this particular case, which is ubiquitous in communication systems.

For the jointly sparse model in~\eqref{eq:jsm}, we define the Hamming distance as
\begin{equation}\label{eq:Hamming}
d_{\text{HD}}(\x_n,\widehat{\x}_n)=\mathbbm{1}_{\x_n\neq \widehat{\x}_n},
\end{equation}
where $\mathbbm{1}_{\mathcal{A}}$ is the indicator function.
If ({\em i}) the pdf $\phi(\x_n)$ in~\eqref{eq:jsm} is a $J$-dimensional Dirac-delta function $\delta(\x_n-\mathbf{1})$, where $\mathbf{1}$ is an all-one row vector, and ({\em ii}) the estimate satisfies $\widehat{x}_n^{\{1\}}=\cdots =\widehat{x}_n^{\{J\}}, \forall n\in \{1,\ldots,N\}$, then the weighted support set error~\eqref{eq:MWSE} with weight $\beta=0.5$ is equal to half of the Hamming distance~\eqref{eq:Hamming} for super symbols.
In~\eqref{eq:MHD}--\eqref{eq:MHD_optTh}, we briefly derive the Hamming distance-optimal estimator when $\x_n\in \{\mathbf{1},\mathbf{0}\}$, where $\mathbf{0}$ is an all-zero row vector. The mean Hamming distance (MHD) given the pseudo data $\q_n$ is
\begin{equation}\label{eq:MHD}
\begin{split}
    &\text{MHD}|\q_n=\mathbb{E}[d_{\text{MHD}}(\x_n,\widehat{\x}_n)|\q_n]=\\
    &\left\{
                \begin{array}{ll}
                 \Pr(\x_n=\mathbf{1}|\q_n),&\   \widehat{\x}_n=\mathbf{0}\ \text{and}\ \x_n=\mathbf{1},\\
                 \Pr(\x_n=\mathbf{0}|\q_n),&\   \widehat{\x}_n=\mathbf{1}\ \text{and}\ \x_n=\mathbf{0},\\
                 0,\quad\quad\quad&\ \widehat{\x}_n= \x_n.
                \end{array}
      		\right.\\
\end{split}
\end{equation}
Following the steps in~\eqref{eq:weighted_optB}--\eqref{eq:b_opt}, the minimum MHD (MMHD) estimator is
\begin{equation}\label{eq:MHD_optTh}
\widetilde{\x}_n=\left\{
                \begin{array}{ll}
                 \mathbf{1},&   \sum_{j=1}^J q_n^{(j)} \geq\frac{J}{2}+\Delta_v \ln\l[\frac{1-\rho}{\rho}\r],\\
                 \mathbf{0},&   \sum_{j=1}^J q_n^{(j)} <\frac{J}{2}+\Delta_v \ln\l[\frac{1-\rho}{\rho}\r].
                \end{array}
      		\right.\\
\end{equation}


\subsection{Mean absolute error (MAE)}\label{sec:MAE}
{\bf MAE-optimal estimator:} The element-wise absolute error (AE) is
\begin{equation}\label{eq:AE}
  d_{\text{AE}}(x_n^{(j)},\widehat{x}_n^{(j)})=|x_n^{(j)}-\widehat{x}_n^{(j)}|.
\end{equation}
In order to find the minimum mean absolute error (MMAE) estimate, $\widetilde{\x}_n$, we need to find the stationary point of~\eqref{eq:AE},
\begin{equation}\label{eq:optimalCondition}
\frac{d\mathbb{E}[|x_n^{(j)}-\widehat{x}_n^{(j)}|\ | \q_n]}{d\widehat{x}_n^{(j)}}\bigg|_{\widehat{x}_n^{(j)}=\widetilde{x}_n^{(j)}}=0,
\end{equation}
$\forall j\in\{1,\ldots,J\},n\in\{1,\ldots,N\}$.
It can be proved that $\mathbb{E}[X]=\int_{0}^{\infty} \Pr(X>x)dx$, if $X\geq 0$. Therefore,
\begin{equation}\label{eq:expectation}
\begin{split}
  &\mathbb{E}[|x_n^{(j)}-\widehat{x}_n^{(j)}|\ | \q_n]=\int_{0}^{\infty} \Pr(|x_n^{(j)}-\widehat{x}_n^{(j)}|>t| \q_n)dt\\
  &=\int_{-\infty}^{\widehat{x}_n^{(j)}} \Pr(x_n^{(j)}<t| \q_n)dt+\int_{\widehat{x}_n^{(j)}}^{\infty} \Pr(x_n^{(j)}>t| \q_n)dt.
\end{split}
\end{equation}
Using~\eqref{eq:optimalCondition} and~\eqref{eq:expectation}, we obtain
\begin{equation*}
  \Pr(x_n^{(j)}<\widetilde{x}_n^{(j)} | \q_n)=\Pr(x_n^{(j)}>\widetilde{x}_n^{(j)}| \q_n).
\end{equation*}
That is,
\begin{equation*}\label{eq:optimalEst}
  \int_{-\infty}^{\widetilde{x}_n^{(j)}} f(x_n^{(j)} | \q_n)dx_n^{(j)}=\int_{\widetilde{x}_n^{(j)}}^{\infty} f(x_n^{(j)} | \q_n)dx_n^{(j)}=\frac{1}{2},
\end{equation*}
through which we solve for the optimal estimator $\widetilde{x}_n^{(j)}$ numerically. 

{\bf Performance limits:} We calculate the MMAE as follows,
\begin{equation}\label{eq:MMAE}
\begin{split}
  \text{MMAE}\!&=\!\mathbb{E}[|\widetilde{x}_n^{(j)}-x_n^{(j)}|]\!=\!\int_{-\infty}^{\infty}\! \mathbb{E}[|\widetilde{x}_n^{(j)}-x_n^{(j)}|\ | \q_n] f(\q_n)d \q_n\\
  &=\int_{-\infty}^{\infty}\Bigg[\int_{-\infty}^{\widetilde{x}_n^{(j)}} -x_n^{(j)} f(x_n^{(j)}| \q_n)d x_n^{(j)}+\\
  &\quad\quad\quad\quad \int_{\widetilde{x}_n^{(j)}}^{\infty} x_n^{(j)}f(x_n^{(j)}| \q_n)d x_n^{(j)}\Bigg] f(\q_n)d \q_n,
\end{split}
\end{equation}
which has to be numerically approximated.

\begin{figure}[t]
\centering
\includegraphics[width=8cm]{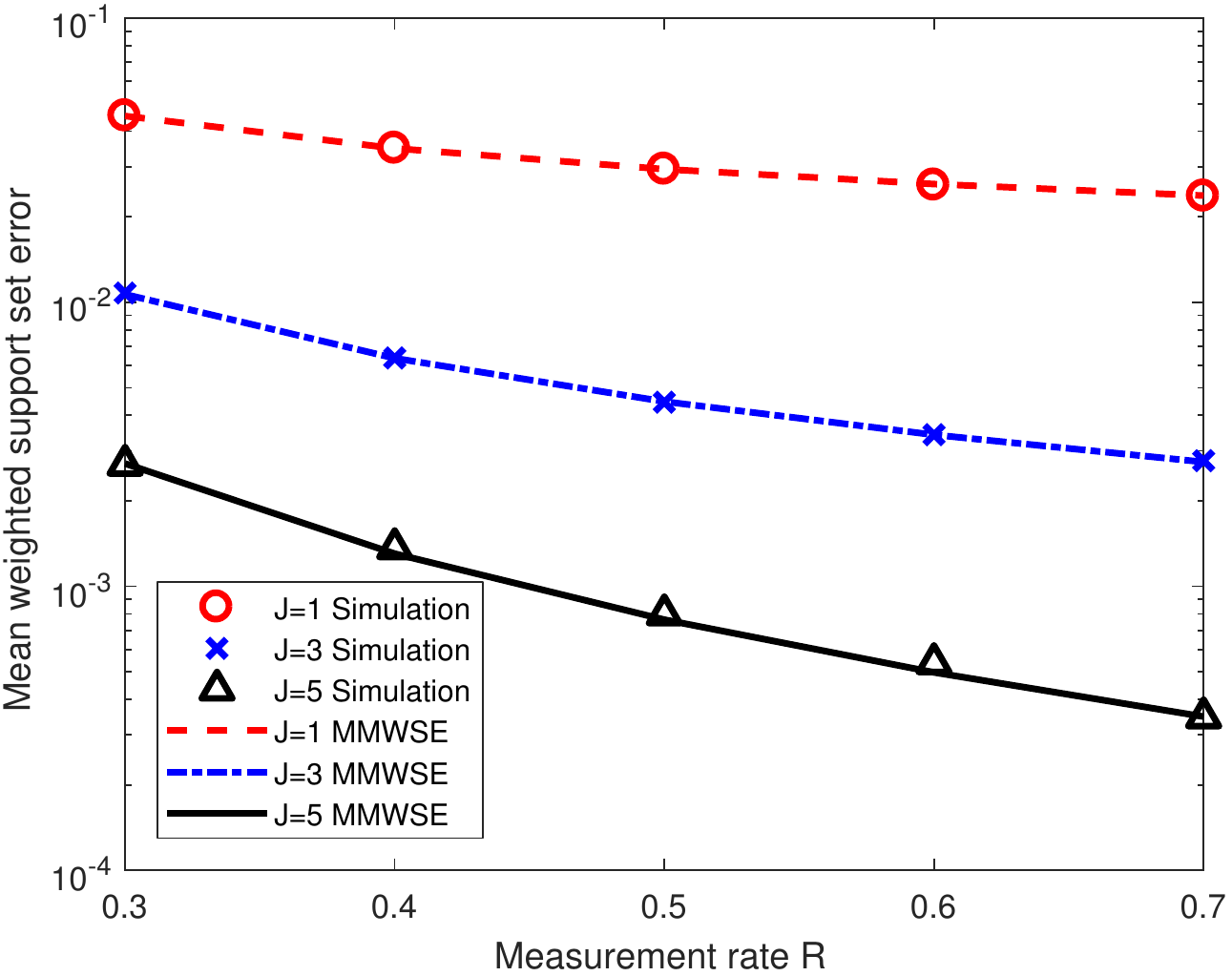}
\caption{Comparison of simulation results to theoretic MMWSE for weighted support set estimation under different number of channels $J$ and measurement rates $R$ ($\beta=0.2$, noise variances $\Delta_z=0.01$).}
\label{fig:SimWSE}
\end{figure}

\section{Synthetic Simulations}\label{sec:numeric_synth}
After deriving the minimum mean weighted support set error (MMWSE) and minimum mean absolute error (MMAE) estimators, this section provides numerical results for Algorithm~\ref{algo:metric_opt_MMV}.
In the case of i.i.d. random matrices and  AWGN channels~\eqref{eq:AWGN}, replica analysis yields the MMSE of the MMV problem~\cite{ZhuBaronKrzakala2017IEEE}. By inverting the MMSE (details in Appendix~\ref{app:inverMMSE}), we obtain the scalar channel noise variance $\Delta_v$,\footnote{Algorithm~\ref{algo:metric_opt_MMV} also applies to problems with other types of matrices, as long as 
the entries in the measurement matrices scale with $\frac{1}{\sqrt{N}}$. However, when the matrices are not i.i.d., there is no easy way to {\em theoretically} characterize the MMSE, the equivalent scalar channel noise variance $\Delta_v$, and the  metric-optimal error. Such a theoretic characterization is sometimes necessary, because the MMSE behaves differently under different noise variances $\Delta_z$~\eqref{eq:AWGN} and measurement rates $R$~\cite{ZhuBaronCISS2013,ZhuBaronKrzakala2017IEEE}.} which characterizes the posterior $f(\x_n|\q_n)$. Given $\Delta_v$, we characterize the MMWSE and MMAE theoretically. In the following simulations, we use i.i.d. Gaussian matrices with unit-norm rows, i.i.d. $J$-dimensional Bernoulli-Gaussian signals~\eqref{eq:jsm} with $J=1,3$, and 5, and sparsity rate $\rho=0.1$. The signal length is $N=10000$, and the measurement rate $R=\frac{M}{N}$ varies from $0.3$ to $0.7$.
For each setting, the simulation results are averaged over 50 realizations of the problem.

\textbf{Mean weighted support set error in AWGN channels:} We simulate AWGN channels~\eqref{eq:AWGN} in this case with the noise variance being $\Delta_z\in \{0.01,0.001\}$. Fig.~\ref{fig:SimWSE} shows the weighted support set estimation results using our metric-optimal algorithm compared to the MMWSE~\eqref{eq:MMWSE_final}. The red dashed curve, the blue dashed-dotted curve, and the black solid curve correspond to the MMWSE of $J=1,\ 3$, and 5, respectively. The red circles, blue crosses, and black triangles represent the simulation results. We can see that our simulation results match the theoretically optimal performance.

\textbf{Remark:} The optimal weighted support set estimator~\eqref{eq:b_opt} is not Lipschitz continuous. Hence, for $J=1$,~\eqref{eq:b_opt} is not guaranteed to yield the MMWSE, according to Lemma~\ref{lemma:optSMV}. Nevertheless, numerical results for $J=1$ show that the MWSE given by~\eqref{eq:b_opt} is close to the MMWSE.

For  weighted support set estimation, we further study the information-theoretic optimal receiver operating characteristic (ROC) curves, which are plotted in Fig.~\ref{fig:ROC} for different $J$'s. The red curves and black curves are 
plotted for noise variances $\Delta_z=0.001$ and $0.01$, respectively. The solid curves, the dashed curves, and the dotted curves represent $J=1,\ 3$, and 5, respectively. The true positive rate (TPR) and false positive rate (FPR) in both axes are defined as TPR$=\frac{\text{\# accurately predicted positives}}{\text{\# all positives in truth}}$ and FPR$=\frac{\text{\# wrongly predicted negatives}}{\text{\# all negatives in truth}}$. We can see that having more signal vectors $J$ leads to a larger area under the ROC curve, which indicates better trade-offs between true positives and false positives.
\begin{figure}[t]
\centering
\includegraphics[width=8cm]{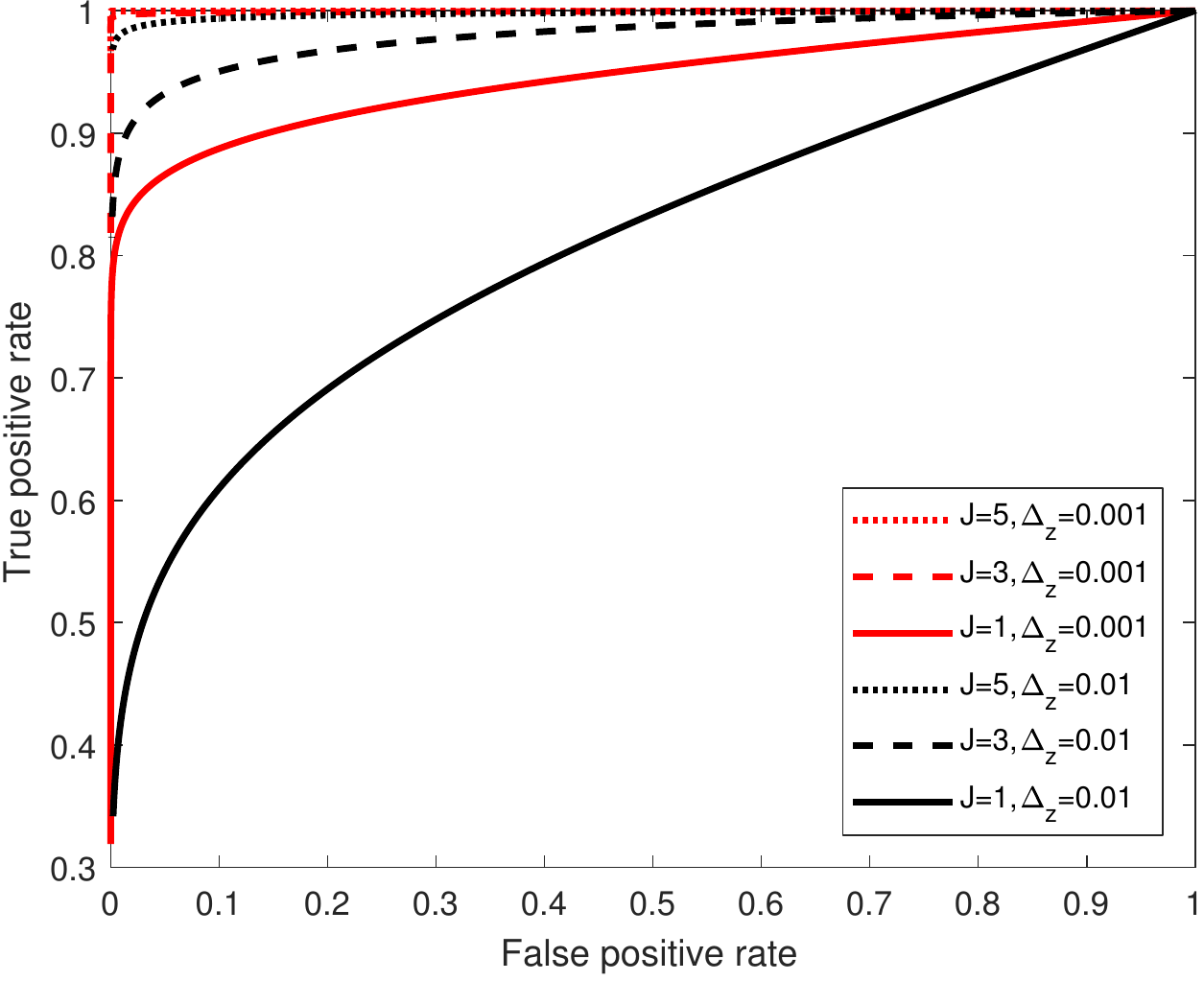}
\caption{Receiver operating characteristic curves for weighted support set estimation under different channel noise variances $\Delta_z$ and number of channels $J$ ($\beta=0.2$, measurement rate $R=0.3$).}
\label{fig:ROC}
\end{figure}

\textbf{Mean absolute error in logistic channels:} We simulate i.i.d. logistic channels~\eqref{eq:logit} with parameters $a=10$ and $30$; the smaller $a$ is, the noisier the channel becomes. Fig.~\ref{fig:MAEvsMMAE} plots the  simulated MAE (crosses) and the theoretic MMAE~\eqref{eq:MMAE} (curves) for various settings.\footnote{We do not have a replica analysis for logistic channels. In order to compute the theoretic MMAE, we use the average $\Delta_v$ from all the 50 simulations for each setting and calculate the MMAE with~\eqref{eq:MMAE}.} Different colors and line shapes refer to different $a$'s and $J$'s, respectively. We can see that our simulation results match the theoretically optimal performance.

\begin{figure}[t]
\centering
\includegraphics[width=8cm]{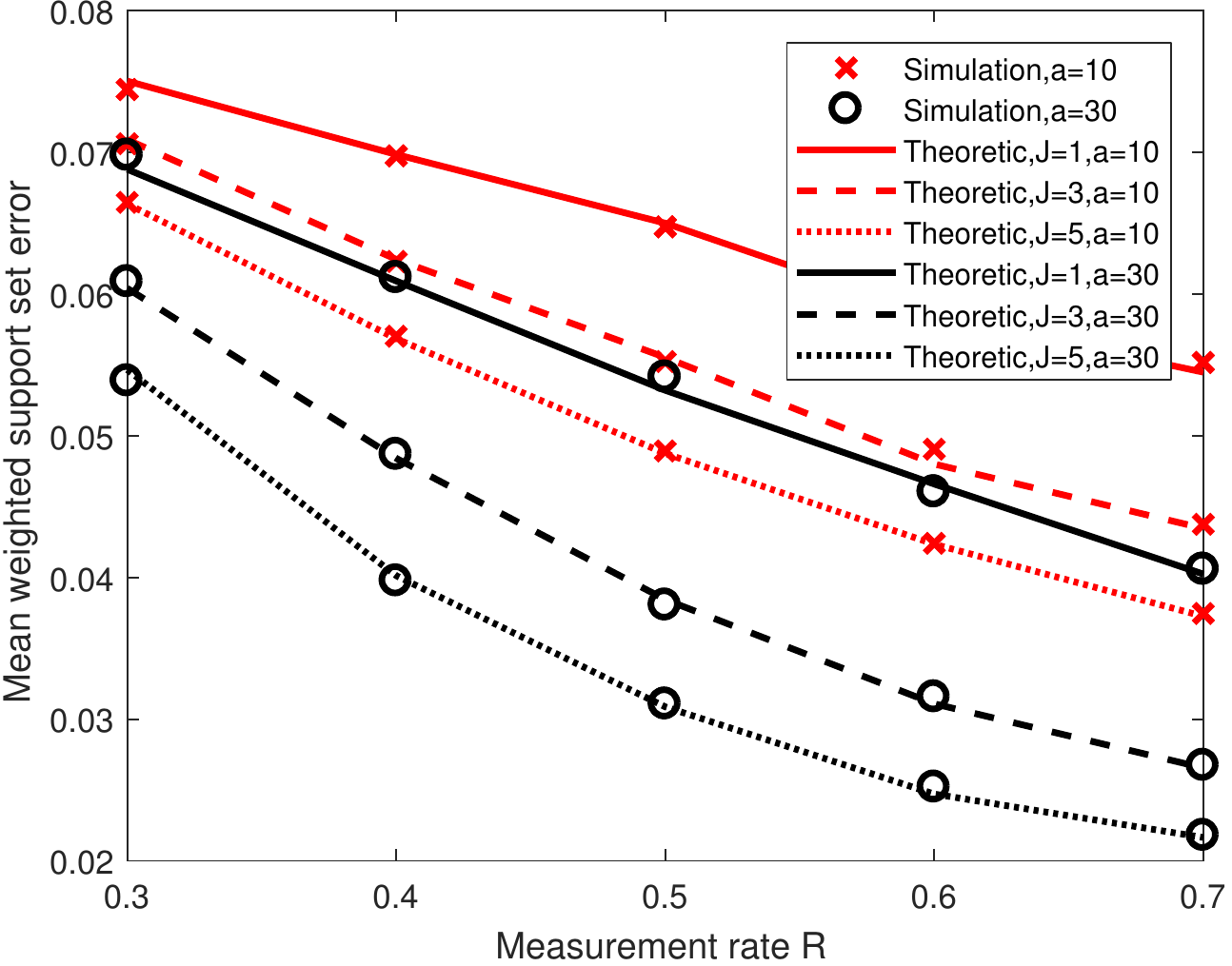}
\caption{Comparison of simulation results to theoretic MMAE  under different logistic channels~\eqref{eq:logit}, number of channels $J$, and measurement rates $R$.}
\label{fig:MAEvsMMAE}
\end{figure}

\textbf{Remark:} Both simulations yield better performance for larger $J$. This is intuitive, because more signal vectors that share the same support should make the estimation process easier due to more information being available.

\begin{figure}[t]
\centering
    \includegraphics[width=8.5cm]{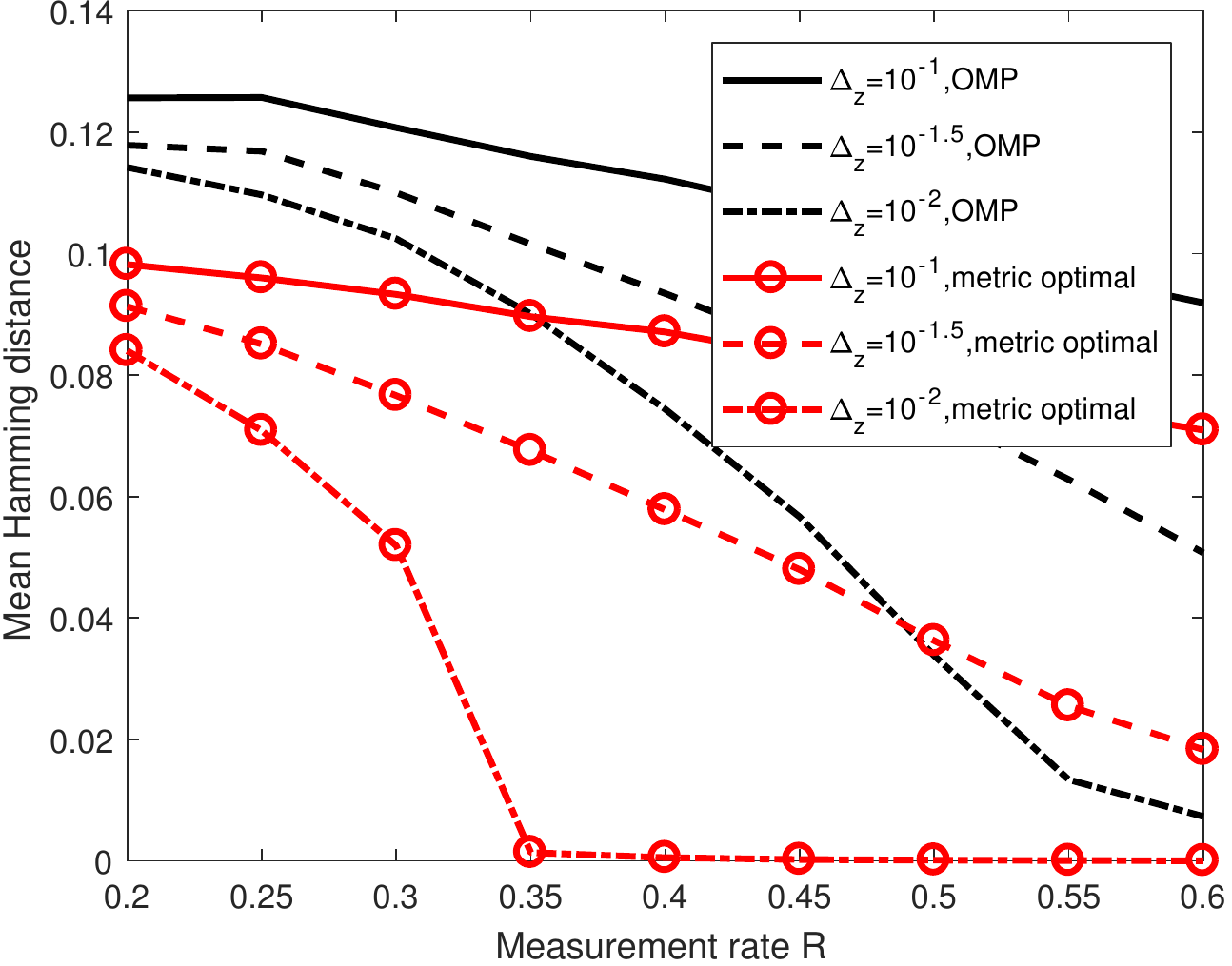}
\caption{Simulation results of OMP and Algorithm~\ref{algo:metric_opt_MMV} for active user detection in multi-user communication (SMV).}\label{fig:OMP}
\end{figure}

\section{Application}\label{sec:app}
In this section, we discuss active user detection (AUD) in a multi-user communication setting that can be viewed as compressed sensing~\cite{DonohoCS,CandesRUP,BaraniukCS2007} (CS, which is closely related to SMV) in some scenarios.  
Next, we simulate AUD using our metric-optimal algorithm. Finally, we discuss how to solve the AUD problem using an MMV setting.

{\bf CS based active user detection:}
One application of MMV is AUD in a massive random access (MRA) scenario for multi-user communication~\cite{FletcherRanganGoyal2009,Boljanovic2017}. 
In the MRA scenario, multiple end users (EUs) are requesting access to the network simultaneously by sending their unique identification codewords, $\a_n\in \{-1,+1\}^{M\times 1},\ n\in \{1,\ldots,N\}$, to the base station, where $n$ denotes the user id, and each user's codeword is known by the base station. The base station needs to determine which EUs are requesting access to the network (active) and which are not (inactive), so that it can allocate resources to the active EUs.
Fletcher et al.~\cite{FletcherRanganGoyal2009} proposed a CS based AUD scheme for MRA, which was recently revisited by  Boljanovi{\'c} et al.~\cite{Boljanovic2017}. In their setup~\cite{FletcherRanganGoyal2009,Boljanovic2017}, all EUs are synchronized, i.e., all active EUs send each entry of their codewords to the base station simultaneously in one time slot. Denote the status of the $n$-th EU by $x_n\in \{0,1\}$, where $x_n=1$ means active and $x_n=0$ is inactive. Denote the received signal vector at the base station by $\y\in \mathbb{R}^{M\times 1}$. Since all EUs are synchronized, we can express the received signal $\y$ by
\begin{equation}\label{eq:MRA_SMV}
\y=\A\x+\z,
\end{equation} 
where $\A=[\a_1,\ldots,\a_N]\in \{0,1\}^{M\times N}$, and $\z$ is AWGN.
The base station estimates $\x$ to determine which EUs are active. 

In the following, we first apply a mean Hamming distance (MHD) optimal algorithm to estimate $\x$ from~\eqref{eq:MRA_SMV}, and compare to the algorithm of Boljanovi{\'c} et al.~\cite{Boljanovic2017}, which is orthogonal match pursuit (OMP)~\cite{Pati1993}. Next, we propose an MMV based scheme for AUD in MRA.

{\bf Simulation with MHD-optimal algorithm:}
As the reader can see, the CS based active user detection~\cite{FletcherRanganGoyal2009,Boljanovic2017} is an SMV problem~\eqref{eq:SMV}, which is MMV for $J=1$~\eqref{eq:linearMixing}. Later in this section we will extend the scheme by Boljanovi{\'c} et al. to an MMV setting, and so we keep using MMV notations. Note that the entries of the measurement matrix in this AUD problem take values of $\pm 1$. Because the derivation of Algorithm~\ref{algo:AMP_MMV} assumes that entries of $\A^{(j)}$ scale with $\frac{1}{\sqrt{N}}$,\footnote{Details can be found in Krzakala et al.~\cite{Krzakala2012probabilistic} and Barbier and Krzakala~\cite{BarbierKrzakala2017IT}.} 
we scale $\A^{(j)}$ down by $\sqrt{N}$ using a modified $\widetilde{\y}^{(j)}=\frac{\y^{(j)}}{\sqrt{N}}$~\eqref{eq:SMV}. 

Following the discussion above, we simulate the settings of measurement rate $R\in \{0.2,0.25,\ldots,0.6\}$ and noise variance $\Delta_z\in \{10^{-1},10^{-1.5},10^{-2}\}$. For each setting, we randomly generate 50 realizations of the Bernoulli signal $\x^{(1)}\in \{0,1\}^{N\times 1}$ with  sparsity rate $\rho=0.1$, and measurement matrix $\A^{(1)}\in \l\{-\frac{1}{\sqrt{N}},+\frac{1}{\sqrt{N}}\r\}^{M\times N}$, where $N=10000$. We run Algorithm~\ref{algo:metric_opt_MMV} with $J=1$ to estimate the underlying signal $\x^{(1)}$. Note that $f_{a_n}(\Delta_v,\q_n)$ and $f_{v_n}(\Delta_v,\q_n)$ in Lines~\ref{line:mean}--\ref{line:var} of Algorithm~\ref{algo:AMP_MMV}, which consists of Part 2 of Algorithm~\ref{algo:metric_opt_MMV}, are given by the following,
\begin{equation*}
\begin{split}
f_{a_n}(\Delta_v,\q_n)&=\frac{\rho}{\rho+(1-\rho)\exp\l[-\frac{\sum_{j=1}^J q_n^{(j)} -\frac{J}{2}}{\Delta_v}\r]}\mathbf{1},\\
f_{v_n}(\Delta_v,\q_n)&=f_{a_n}(\Delta_v,\q_n)-f_{a_n}(\Delta_v,\q_n)^2,
\end{split}
\end{equation*}
where the power-of-two in the last term of $f_{v_n}(\cdot,\cdot)$ is applied element-wise.

Our results are compared to OMP in Fig.~\ref{fig:OMP}. The solid, dashed, and dash-dotted curves represent noise variance $\Delta_z=10^{-1},\ 10^{-1.5},\ 10^{-2}$, respectively. The black curves are the OMP results and the red curves with circle markers are the results of Algorithm~\ref{algo:metric_opt_MMV} when optimizing for MHD. We can see that our algorithm consistently outperforms OMP.\footnote{Note that the entries of the signal estimated by OMP are not exactly 0's and 1's. Hence, in order to provide meaningful results, we threshold the OMP estimates before calculating the Hamming distance.}

{\bf MMV scheme for active user detection:}
Reminiscing on Section~\ref{sec:numeric_synth} and our previous work with Krzakala on MMV~\cite{ZhuBaronKrzakala2017IEEE}, more measurement vectors (larger $J$) lead to better estimation quality. We propose to convert the SMV style of the AUD problem into an MMV style by having each EU send $J$ different identification codewords, $\a_n^{(j)}\in \{-1,+1\}^{M\times 1}, \forall j\in \{1,\ldots,J\}$, to the base station. 
However, since the underlying signal $\x$~\eqref{eq:MRA_SMV} is Bernoulli and does not change during the AUD period, the resulting MMV is equivalent to an SMV with $J$ times more measurements; the column $\a_n$ of the measurement matrix in the equivalent SMV scheme is $\a_n = \l[\l(\a_n^{(1)}\r)^T,\cdots, \l(\a_n^{(J)}\r)^T\r]^T$. Hence, Algorithm~\ref{algo:metric_opt_MMV} should yield the same detection accuracy for the MMV scheme with $J$ channels  and the SMV scheme with $J$ times larger measurement rate. Nevertheless, there is one advantage in adopting the MMV scheme: Lines~\ref{line:beginForLoop}--\ref{line:endForLoop} of Algorithm~\ref{algo:AMP_MMV}\footnote{Recall that Algorithm~\ref{algo:metric_opt_MMV} runs Algorithm~\ref{algo:AMP_MMV} as a subroutine.} can be parallelized with $J$ processing units (for example using a general purpose graphics processing unit or multicore computing system). After parallelizing Algorithm~\ref{algo:AMP_MMV}, the base station can perform the detection procedure with less runtime.

\begin{figure*}[bt]
\begin{equation}\label{eq:Ez}
\mathbb{E}[w|k,y,\Theta]=\left\{
                \begin{array}{ll}
                \displaystyle \frac{1}{\widetilde{Z}}\int  \frac{w}{1+\operatorname{e}^{-aw}} \frac{\operatorname{e}^{-\frac{1}{2\Theta}(w-k)^2}}{\sqrt{2\pi\Theta}}  dz=\frac{\sum_{u=1}^{u_{\text{max}}} \alpha_u T_1(u)}{\sum_{u=1}^{u_{\text{max}}} \alpha_u T_0(u)}=k+\frac{\Theta\sum_{u=1}^{u_{\text{max}}} \frac{\alpha_u \phi(\eta_u)}{\sqrt{\l(\frac{\sigma_u}{a}\r)^2+\Theta}}}{\sum_{u=1}^{u_{\text{max}}} \alpha_u \Phi(\eta_u)},\ &y=1,\\
               \displaystyle  \frac{1}{\widetilde{Z}}\int \frac{w\operatorname{e}^{-aw}}{1+\operatorname{e}^{-aw}} \frac{\operatorname{e}^{-\frac{1}{2\Theta}(w-k)^2}}{\sqrt{2\pi\Theta}}  dz=\frac{k-\sum_{u=1}^{u_{\text{max}}} \alpha_u T_1(u)}{1-\sum_{u=1}^{u_{\text{max}}} \alpha_u T_0(u)}=k-\frac{\Theta\sum_{u=1}^{u_{\text{max}}} \frac{\alpha_u \phi(\eta_u)}{\sqrt{\l(\frac{\sigma_u}{a}\r)^2+\Theta}}}{1-\sum_{u=1}^{u_{\text{max}}} \alpha_u \Phi(\eta_u)},\ &y=0.
                \end{array}
      		\right.\\
\end{equation}
\end{figure*}

\begin{figure*}
\begin{equation}\label{eq:mmse_p1}
\begin{split}
\mathbb{E}[(\mathbb{E}[\x_n|\q_n])^2]&=\int_{\q_n} f(\q_n) (\mathbb{E}[\x_n|\q_n])^2 d\q_n\\
&=\frac{\l(\frac{\rho}{1+\Delta_v}\r)^2}{\l[2\pi (1+\Delta_v)\r]^{J/2}}\int_{\q_n} \frac{\q_n \q_n^T}{\rho \exp\l[\frac{\q_n\q_n^T}{2(1+\Delta_v)}\r]+(1-\rho)\l(1+\frac{1}{\Delta_v}\r)^{J/2} \exp\l[\frac{\q_n\q_n^T (\Delta_v-1)}{2\Delta_v(1+\Delta_v)}\r]}  d\q_n.
\end{split}
\end{equation}
\end{figure*}

\section{Conclusion}\label{sec:conclusion}
In this paper, we studied the MMV signal estimation problem with user-defined additive error metrics on the estimate. We proposed an  algorithmic framework that is optimal under  arbitrary additive error metrics. We showed the  optimality of our algorithm under certain conditions for SMV and conjectured its optimality for MMV. As examples, we derived algorithms that yield the optimal estimates in the sense of mean weighted support set error and mean absolute error, respectively. Numerical results not only verified the theoretic performance but also verified the intuition that having more signal vectors in MMV problems is beneficial to the estimation algorithm.
We further provided simulation results for active user detection problem in multi-user communication systems, which is a real-world application of MMV  models with the goal of minimizing the Hamming distance. Simulation results demonstrated the promise of our algorithm.

\appendix

\subsection{Derivation of $g_{\text{out}}$ for logistic channels~\eqref{eq:logit}}\label{app:logit}
Byrne and Schniter~\cite{ByrneSchniter2015ArXiv} provide a method to derive $g_{\text{out}}$ for logistic channels~\eqref{eq:logit}, but the actual formula for $g_{\text{out}}$ is not given in their paper. To make our paper self-contained, we outline the derivation of $g_{\text{out}}$ for logistic channels. 
In order to calculate $g_{out}(k,y,\Theta)$~\eqref{eq:g_out}, we need
to find $f(w|k,y,\Theta)$~\eqref{eq:prob_w} and calculate $\mathbb{E}[w|k,y,\Theta]$. For logistic channels~\eqref{eq:logit}, 
\begin{equation}\label{eq:logit_cond}
\begin{split}
& f(w|k,y,\Theta)=\\
&\frac{1}{\widetilde{Z}}\times \l[\frac{\delta(y-1)}{1+\operatorname{e}^{-w}}+\delta(y)\frac{\operatorname{e}^{-w}}{1+\operatorname{e}^{-w}}\r]\frac{1}{\sqrt{2\pi\Theta}} \operatorname{e}^{-\frac{1}{2\Theta}(w-k)^2},
\end{split}
\end{equation}
where $\widetilde{Z}$ is a normalization factor. Therefore, it is
difficult to calculate $\mathbb{E}[w|k,y,\Theta]$. Instead of calculating $\mathbb{E}[w|k,y,\Theta]$ by brute force, Byrne and Schniter~\cite{ByrneSchniter2015ArXiv}  use a mixture of Guassian cumulative distribution functions (CDF's) to approximate the sigmoid function $\frac{1}{1+\text{exp}(-aw)} \approx\sum_{u=1}^{u_{\text{max}}} \alpha_u \Phi(\frac{w}{\sigma_u/a})$~\cite{Stefanski1991}, where $u_{\text{max}}$ is the maximum number of Gaussian CDF's one wants to use, $\Phi(\frac{w}{\sigma_u/a})$ denotes the Gaussian CDF whose standard deviation is $\frac{\sigma_u}{a}$, and $\alpha_u$ is the weight.

Following Byrne and Schniter~\cite{ByrneSchniter2015ArXiv}, we define the $i$-th moment
\begin{equation*}
\int w^i \mathcal{N}(w;k,\Theta) \Phi\l(\frac{w}{\sigma_u/a}\r)dz=T_i(u),
\end{equation*}
where $\mathcal{N}(w;k,\Theta)$ is the pdf of an RV $w$ with mean $k$ and variance $\Theta$, and $\Phi(\cdot)$ is the CDF of a standard Gaussian RV.
Defining $\eta_u=\frac{k}{\sqrt{\l(\frac{\sigma_u}{a}\r)^2+\Theta}}$, we obtain
\begin{eqnarray}
T_0(u)&=&\Phi(\eta_u),\nonumber\\
T_1(u)&= &k\Phi(\eta_u)+\frac{\Theta\phi(\eta_u)}{\sqrt{\l(\frac{\sigma_u}{a}\r)^2+\Theta}},\nonumber
\end{eqnarray}
\begin{eqnarray}
T_2(u)=\frac{(T_1(u))^2}{\Phi(\eta_u)}+\Theta \Phi(\eta_u)-\frac{\Theta^2\phi(\eta_u)}{\l(\frac{\sigma_u}{a}\r)^2+\Theta}\l(\eta_u+\frac{\phi(\eta_u)}{\Phi(\eta_u)}\r),\nonumber
\end{eqnarray}
where $\phi(\eta_u)$ is the pdf of a standard Gaussian RV at $\eta_u$.
Hence, the normalization factor $\widetilde{Z}$ in~\eqref{eq:logit_cond} can be derived,
\begin{equation*}
\begin{split}
\widetilde{Z}=& \int \l[\frac{\delta(y-1)}{1+\operatorname{e}^{-aw}}+\delta(y)\frac{\operatorname{e}^{-aw}}{1+\operatorname{e}^{-aw}}\r]\frac{\operatorname{e}^{-\frac{1}{2\Theta}(w-k)^2}}{\sqrt{2\pi\Theta}}  dz\\
 = &\left\{
                \begin{array}{ll}
                 \sum_{u=1}^{u_{\text{max}}} \alpha_u \Phi(\eta_u),\ &y=1,\\
                 1-\sum_{u=1}^{u_{\text{max}}} \alpha_u \Phi(\eta_u),\ &y=0.
                \end{array}
      		\right.\\
\end{split}
\end{equation*}
We can further obtain the expression for $\mathbb{E}[w|k,y,\Theta]$ in~\eqref{eq:Ez}.
Hence, we can calculate $g_{\text{out}}$~\eqref{eq:g_out}.

Apart from $g_{\text{out}}$, we also need to find the partial derivative of $g_{\text{out}}$~\eqref{eq:deri_g_out}, which 
according to Rangan~\cite{RanganGAMP2011ISIT} satisfies
\begin{equation}\label{eq:deri_g_out}
-\frac{\partial}{\partial k}g_{\text{out}}(k,y,\Theta)=\frac{1}{\Theta}\l(1-\frac{\text{var}(w|k,y,\Theta)}{\Theta}\r),
\end{equation}
where $\text{var}(w|k,y,\Theta)=\mathbb{E}[w^2|k,y,\Theta]-[\mathbb{E}[w|k,y,\Theta]]^2$. Note that $\mathbb{E}[w^2|k,y,\Theta]$ can be derived in the same way as~\eqref{eq:Ez} and the result is given below,
\begin{equation*}
\mathbb{E}[w^2|k,y,\Theta]=\left\{
                \begin{array}{ll}
                  \displaystyle k^2+\Theta+\frac{ \displaystyle\sum_{u=1}^{u_{\text{max}}} \alpha_u \xi_u}{\displaystyle\sum_{u=1}^{u_{\text{max}}} \alpha_u \Phi(\eta_u)},\ &y=1,\\
                 \displaystyle k^2+\Theta-\frac{\displaystyle\sum_{u=1}^{u_{\text{max}}} \alpha_u \xi_u}{\displaystyle 1-\sum_{u=1}^{u_{\text{max}}} \alpha_u \Phi(\eta_u)},\ &y=0,
                \end{array}
      		\right.\\
\end{equation*}
where
\begin{equation*}
\xi_u=\frac{2k\Theta\phi(\eta_u)}{\sqrt{\l(\frac{\sigma_u}{a}\r)^2+\Theta}}-\frac{\Theta^2 \eta_u\phi(\eta_u)}{\l(\frac{\sigma_u}{a}\r)^2+\Theta}.
\end{equation*}

\subsection{Inverting the MMSE}\label{app:inverMMSE}
For the MMV problem with i.i.d. matrices and joint Bernoulli-Gaussian signals, Zhu et al. provide an information-theoretic characterization of the MMSE by using replica analysis~\cite{ZhuBaronKrzakala2017IEEE}. Suppose that we have already obtained the MMSE for an MMV problem. This appendix briefly shows 
how to invert the MMSE expression in order to obtain the equivalent scalar channel noise variance $\Delta_v$.

The optimal denoiser for the pseudo data is
$\widetilde{\x}_n=\mathbb{E}[\x_n|\q_n]=f_{a_n}(\Delta_v,\q_n)$,
where $f_{a_n}(\Delta_v,\q_n)$ is given in~\eqref{eq:denoiser}. We then express MMSE expression using $\mathbb{E}[\x_n|\q_n]$ as follows,
\begin{equation}
\text{MMSE}=\mathbb{E}[(\widetilde{\x}_n-\x_n)^2]=J\rho-\mathbb{E}[(\mathbb{E}[\x_n|\q_n])^2].\label{eq:mmse1}
\end{equation}
We calculate $\mathbb{E}[(\mathbb{E}[\x_n|\q_n])^2]$ in~\eqref{eq:mmse_p1},
where the $J$-dimensional integral can be simplified by a change of coordinates. Then, we plug~\eqref{eq:mmse_p1} into~\eqref{eq:mmse1}, and express the MMSE as a function of $\Delta_v$. Finally, we numerically solve $\Delta_v$ for any given MMSE.

\section*{Acknowledgments}
We would like to thank Jin Tan for providing valuable insights into achieving metric-optimal performance in signal estimation. Jong Chul Ye, Yavuz Yapici,
and Ismail Guvenc helped us identify some real-world applications for minimizing error metrics that are different from the MSE. Finally, we are very grateful to the reviewers and Associate Editor Prof. Lops. In addition to their excellent suggestions, they were unusually flexible with us during the review process.


\end{document}